\documentclass{article}
	\usepackage{amsmath,amssymb,amsfonts,amsthm}
	\usepackage{graphicx}
	\usepackage{cite}
	\usepackage{fullpage}
	\usepackage{commath}
	\newtheorem{theorem}{Theorem}[section]
	\newtheorem{lemma}[theorem]{Lemma}

	\newtheorem{claim}{Claim}[theorem]
	\theoremstyle{definition}

	\theoremstyle{remark}

		\title{On the complexity of the outer-connected bondage and the outer-connected reinforcement problems}
		\author{ M. Hashemipour$^1$,  M. R. Hooshmandasl$^2$, A. Shakiba$^3$ \\
			\footnotesize{$^{1,2}$Department of Computer Science, Yazd University, Yazd, Iran. }  \\
			\footnotesize{$^{1,2,3}$The Laboratory of Quantum Information Processing, Yazd University, Yazd, Iran.}  \\
			\footnotesize{$^{3}$Department of Computer Science, Vali-e-Asr University of Rafsanjan, Rafsanjan, Iran}\\
			\footnotesize{e-mail: $^1$mhashemi@stu.yazd.ac.ir, $^2$hooshmandasl@yazd.ac.ir, $^3$ali.shakiba@vru.ac.ir.
			}}
			
			\date{}

\begin{document}
			\maketitle
			
	\begin{abstract}
	Let $G=(V,E)$ be a graph. 	A subset $S \subseteq V$ is a \emph{dominating set} of
	$G$ if every vertex not in $S$ is adjacent to a vertex in $S$.  	A set $\tilde{D} \subseteq V$ of a graph $G=(V,E) $ is called an \emph{outer-connected dominating set} for $G$ if (1) $\tilde{D}$ is a dominating set for $G$, and (2)  $G [V \setminus \tilde{D}]$, the induced subgraph of $G$ by $V \setminus \tilde{D}$, is connected. The minimum size among all outer-connected dominating sets of $G$ is called the \emph{outer-connected domination number} of $G$ and is denoted by $\tilde{\gamma}_c(G)$. We define the \emph{outer-connected  bondage number} of a graph $G$ as the minimum number of edges
	whose removal from $G$ results in a graph with an outer-connected domination number larger than the one for $G$. Also, the \emph{outer-connected reinforcement
		number} of a graph $G$   is defined as  the minimum number of edges whose addition to $G$ results
	in a graph with an outer-connected domination number, which is  smaller than the one for $G$. This paper shows that the decision problems for the \emph{outer-connected bondage } and the 
	\emph{outer-connected reinforcement} numbers are $\mathbf{NP}$-hard. Also, the exact values of the bondage number are determined for several classes of graphs.
\end{abstract}	
	\section{Introduction}	
    	The terminology and notation on graph theory in this paper in general follows the reference \cite{haynes1998fundamentals}.
		Let $G = (V,E)$ be a graph with vertex set $V$ and edge set $E$. The graph $G$ is called to be of  order $|V|$  and size $|E|$.  Also, we use $V(G)$ and $E(G)$ to denote the vertex set and 
		the edge set for the graph $G$, respectively.
		Let  $v$ be a vertex in $V.$ The \emph{open neighborhood} of $v$ is denoted by $N_G(v)$ and is defined as $ \{u \in V : \{u,v\} \in
		E(G)\}$. Similarly, the \emph{closed neighborhood} of $v$ is denoted by $N_G[v]$ and is defined as $\{v\} \cup N_G(v)$.  Whenever the graph $G$
		is clear from the context, we simply write $N(v)$ and $N[v]$ to denote $N_G(v)$ and $N_G[v]$, respectively.
	A \emph{leaf} vertex in $G$ is a vertex of degree one. We denote the \emph{path} of order $n$ by $P_n$, the \emph{cycle} of order $n$ by $C_n$ and the \emph{star}
		 of order $n$ by $S_n$.
	A forest where each component is a star is called a \emph{galaxy}.
	 For a subset $S$ of vertices of $G$, we refer to $G[S]$ as the subgraph of $G$ induced by $S$.
		A subset $S \subseteq V$ is a \emph{dominating set} of
		$G$ if every vertex not in $S$ is adjacent to a vertex in $S$. The \emph{domination number}
		of $G$, denoted by $\gamma(G)$, is the minimum cardinality among all dominating sets of $G$. A
		dominating set $S$ is called a $\gamma-set$ of $G$ if $|S| =\gamma(G)$.
		
		 Domination is  one of the most widely studied topics in graph theory, e.g. \cite{haynes1998fundamentals,teresa1998domination} and refernces therein.
				This paper studies some issues in a particular variation of domination, namely the \emph{outer-connected domination}.
			  The concept of outer-connected domination number is introduced by Cyman in \cite{cyman2007outer1connected} and is further studied by others in \cite{akhbari2013outer,keil2013computing}. The outer-connected domination problem is shown to be an  $\mathbf{NP}$-complete problem for
	arbitrary graphs in \cite{cyman2007outer1connected}.
	A set $\tilde{D} \subseteq V$ of a graph $G=(V,E) $ is called an \emph{outer-connected dominating set} for $G$ if (1) $\tilde{D}$ is a dominating set for $G$, and (2)  $G [V \setminus \tilde{D}]$, the induced subgraph of $G$ by $V \setminus \tilde{D}$, is connected. The minimum size among all outer-connected dominating sets of $G$ is called the \emph{outer-connected domination number} of $G$ and is denoted by $\tilde{\gamma}_c(G)$ \cite{cyman2007outer1connected}.
	An outer-connected	dominating set $\tilde{D}$ is called a $\tilde{\gamma}-set$ of $G$ if $|\tilde{D}| =\tilde{\gamma}_c(G)$. 
	
	
In this paper, we focus on two graph alterations and their effects on the outer-connected domination number, (1) the removal of edges  from a graph and (2) the addition of edges to a graph.
	The bondage  and the reinforcement numbers are two important parameters for measuring the
	vulnerability and the stability of the network domination under link failure and link addition.
   The \emph{bondage number} of $G$, denoted by $b(G)$, is the minimum
number of edges whose removal from $G$ results in a graph with a domination number larger than the one for $G$. The
\emph{reinforcement number} of $G$, denoted by $r(G)$, is the smallest number of edges whose addition to G results in a graph with a
domination number smaller than the one for $G$. The bondage  and the reinforcement numbers  in graphs are very interesting research problems and  were introduced by Fink et al. in \cite{fink1990bondage} and Kok, Mynhardt in \cite{kok1990reinforcement}, respectively. Hattingh et al. in \cite{hattingh2008restrained} showed that the problem of  the restrained bondage  is $\mathbf{NP}$-complete, even
for bipartite graphs. Also, he has determined the exact values of the bondage number for several classes of graphs. Moreover, the reinforcement number for digraphs has been studied by Huang, Wang and Xu in \cite{huang2009reinforcement}.  Hu and Xu in \cite{hu2012complexity}  showed that the problems of the bondage, the total bondage, the reinforcement and the total
reinforcement numbers for an arbitrary graph are all $\mathbf{NP}$-hard, in  general. Recently, Xu in \cite{xu2013bondage} gave a review article on the bondage numbers.
 Moreover, Hu  and Sohn  in \cite{hu2014algorithmic}  showed that these problems remain  $\mathbf{NP}$-hard, even for bipartite graphs.
 Xu, Hu and Lu in \cite{lu2015p} studied the complexity of $p$-reinforcement and paired bondage problems in general graphs.
Jafari Rad in \cite{rad2016complexity} showed that  the problems of the $p$-reinforcement, the $p$-total reinforcement, the total restrained reinforcement and the $k$-rainbow reinforcement are all $\mathbf{NP}$-hard for bipartite graphs.  In addition, he also in \cite{rad2017complexity} showed  that the problems of the  paired bondage, the total restrained bondage, the independent bondage and the  $k$-rainbow bondage numbers  are all $\mathbf{NP}$-hard, even if they are restricted to bipartite graphs.  
From the algorithmic point of view, Hartnell et al. in \cite{hartnell1998edge} designed a linear time algorithm to compute the bondage
number of a tree.

The \emph{outer-connected  bondage number} of a graph $G$, where $G$ does not have any  isolated vertex, is denoted by $b_{OCD} (G)$ and is equal to the minimum number of edges
whose removal from $G$ results in a graph with an outer-connected domination number larger than the one for $G$. The \emph{outer-connected reinforcement
	number} of a graph $G$ which does not have any isolated vertices is denoted by $r_{OCD} (G)$ and  is equal to  the smallest number of edges whose addition from $G$ results
in a graph with an outer-connected domination number smaller than the one for $G$.

The rest of the paper is organized as follows: In Section 2, we describe some necessary preliminaries. In Sections 3 , we show  that the decision problem for the  outer-connected reinforcement number in general  graphs is $\mathbf{NP}$-hard. In sections 4, we show that the  outer-connected bondage number is also  $\mathbf{NP}$-hard in general  graphs. In the other words, we show that there are no polynomial time algorithms to compute these values for graphs, unless $P=NP$.  Finally, in Section 5, we determined the exact value of  bondage number for several classes of graphs.

\section{Preliminaries}

In order to show the $\mathbf{NP}$-hardness of the aforementioned problems, we do a polynomial time reduction from 3-satisfiability problem, 3-SAT, which is known to be an $\mathbf{NP}$-complete problem \cite{garey2002computers}. For concreteness,  Let $\mathcal{U}$ be a set of Boolean variables. A \emph{truth assignment} for $\mathcal{U}$ is a mapping $f : \mathcal{U} \rightarrow \{T,F\} $. If $f(u) = T$, then $u$ is said to be $ ``true"$ with respect to  $f$. In the case that  $f(u) = F$, then $u$ is said to be $``false"$ with respect to $f$. If $u$ is a variable in $\mathcal{U}$, then $u$ and $\bar{u}$ are literals over $\mathcal{U}$. The
 literal $u$ is true  if and only if the variable $u$ is true with respect to $f$ and the literal $\bar{u}$ is true if and only if the variable $u$ is false with respect to  $f$ .
 
 A \emph{clause} over $\mathcal{U}$ is a set of literals over $\mathcal{U}$ which represents the disjunction of these literals. It is said to be satisfied by a truth assignment if and only if at least one of its members is true with respect to that assignment. Similary, a  collection $\mathcal{C} = \set{C_1,C_2,\dots,C_m}$ of clauses over $\mathcal{U}$ is
 satisfiable if and only if there exists some truth assignment for $\mathcal{U}$, which  simultaneously satisfies all the clauses $C_i$ in $\mathcal{C}$ for $i = 1, 2, \dots,m$. Such a truth assignment is called a satisfying truth assignment for $\mathcal{C}$. Given these notations, the 3-SAT problem is specified as follows:
 
 \emph{3-SAT problem:}
 
\emph{ Instance: A collection $\mathcal{C} = \set{C_1,C_2,\dots,C_m}$ of clauses over a finite set  of variables $\mathcal{U}$ such that $|C_j| = 3$
 for $j = 1, 2, \dots,m$.}
 
\emph{ Question: Is there a truth assignment for $\mathcal{U}$ which satisfies all the clauses in $\mathcal{C}$?}
 
 \begin{theorem}
 	 (See Theorem 3.1 in \cite{cook2000p}.) The 3-SAT problem is $\mathbf{NP}$-complete.
\end{theorem}
	
\section{$\mathbf{NP}$-hardness for the outer-connected reinforcement problem}	
In this section, we show that the outer-connected reinforcement problem for general graphs, is an NP-hard problem. The outer-connected reinforcement problem is defined as follows:

Outer-connected reinforcement problem:

Instance: A graph $G$ with no isolated vertices and a positive integer k.

Question: Is $ r_{OCD}(G) \le k $?

\begin {theorem}
The outer-connected reinforcement problem is NP-hard.
\end{theorem}
\begin{proof}
	We show the NP-hardness of the outer-connected reinforcement problem by a polynomial reduction from the 3-SAT.
	
	Let $I = (\mathcal{U} = \set{u_1,u_2,\dots,u_n}$ , $\mathcal{C} = \set{C_1,C_2,\dots,C_m})$ be an arbitrary instance of
	the 3-SAT problem. Without loss of generality, consider that $k=1$. We  construct a graph $G$  such that this instance of 3-SAT
	will be satisfiable if and only if $G$ has an outer-connected reinforcement of cardinality equal to  1, i.e. $ r_{OCD}(G) = 1 $. 
	Next, we describe the construction of $G$.
	
	To each $u_i \in \mathcal{U}$, we associate a triangle $\mathcal{S}_i= \set{u_i,v_i,\bar{u_i}}$. Note that $v_i \notin \mathcal{U}$. For each  clause $C_j  \in \mathcal{C}$, we associate a single vertex $c_j$ and add edges $\{c_j,u_i\}(\{c_j,\bar{u_i}\})$ if the literal $u_i(\bar{u_i}) $ appears in clause $C_j$, for $ j = 1, 2,\cdots,m$, respectively. Finally, we add vertices $x$ and $y$ and join them to every vertex $c_j$  for $ j = 1, 2,\cdots,m$  and add edges $\{x,y\}$, $\{y,u_i\}$ and $\{y,\bar{u_i}\}$ for $ i = 1, 2,\cdots,n$.

		\begin{figure}
\centering
\includegraphics[width=0.7\linewidth]{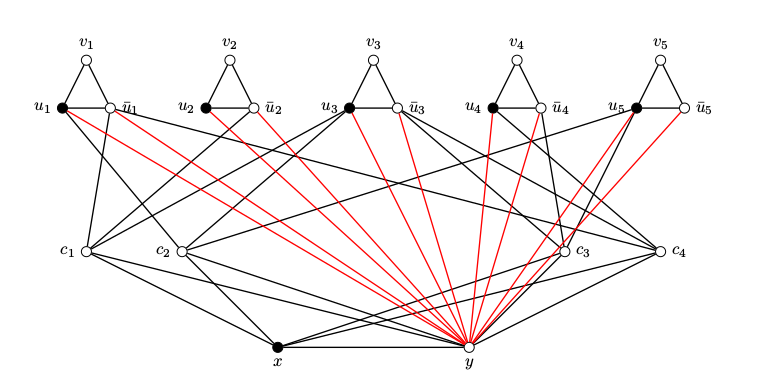}
\caption{{\footnotesize An instance of the outer-connected reinforcement problem. Bold points are the dominator vertices, $k=1$ and $\tilde{\gamma}_c=6.$ }}
\label{fig1}
\end{figure}

	For example, consider a 3-SAT instance $(\mathcal{U} = \set{u_1,u_2,u_3,u_4,u_5}$ , $\mathcal{C} = \set{C_1,C_2,C_3,C_4})$, where $C_1=\{\bar{u_1},\bar{u_2},u_3\},C_2=\{u_1,u_3,u_5\},C_3=\{\bar{u_3},\bar{u_4},u_5\}$ and $C_4=\{\bar{u_1},\bar{u_3},u_4\}$. Figure \ref{fig1} illustrates  the constructed graph corresponding to this instance.
	
	It can be  easily seen that the construction can be accomplished in polynomial time, since the graph $G$ contains $3n+m+2$ vertices and $5n+5m+1$ edges. All that remains to be shown is 
	that the $I = \mathcal{(U,C)}$ is satisfiable if and only if $r_{OCD}(G) = 1$. To this end,
	we will first show the following three claims.
	
\begin{quote}	
	\begin{claim} 
		\label{claim31}
	
		For any graph $G$ constructed as  is described above, we have $\tilde{\gamma}_c(G) = n+1$.
		\end{claim}
		
		\begin{proof}
			Let $\tilde{D}$ be a 	$\tilde{\gamma}$-set of $G$. Then 	$\tilde{\gamma}_c(G) = |\tilde{D}| \ge n+1$ 
			since it is necessary that $|\tilde{D} \cap V(\mathcal{S}_i)| \ge 1$ for  $i=1,2,\dots,n$ and also $|\tilde{D} \cap N[x]| \ge 1$. On the other hand, the set $\tilde{D}^\prime = \{x,u_1,u_2,\dots,u_n\}$ is an outer-connected dominating set for $G$, which implies that 
				$\tilde{\gamma}_c(G) \le |\tilde{D}^\prime| = n+1$.
				Thus, we obtain	$\tilde{\gamma}_c(G) =  n+1.$
		\end{proof}

\begin{claim}
	\label{claim32}
	Let $\tilde{D}_e$ denotes a $\tilde{\gamma}$-set of $G+e$ for an arbitrary edge  $e \in E(\bar{G})$. If there exists an edge $e \in E(\bar{G})$ such that $\tilde{\gamma}_c(G+e) = n$, then for  $i=1,2,\dots,n$, we have $|\tilde{D}_e \cap V(\mathcal{S}_i)| =1$  such that $c_j \notin \tilde{D}_e$ for $j=1,2,\dots,m$  and 
	$y \notin \tilde{D}_e$.
\end{claim}

\begin{proof}
	Since the connection between  the vertices in $V \setminus \tilde{D}_e$ is due to vertex $y,$ then 	$y \notin \tilde{D}_e$.
	On the contrary, suppose that  $|\tilde{D}_e \cap V(\mathcal{S}_\ell)| =0$ for some $\ell=1,2,\dots,n$.
	Since $v_\ell$ needs to be dominated by vertices in $\tilde{D}_e$ and $v_\ell,u_\ell,\bar{u}_\ell \notin \tilde{D}_e$,  
	then one  of the end-vertices of the edge $e$ should be  $v_\ell$, otherwise $\tilde{D}_e$ dominates it via
	the edge $e$ in $G+e$ and for every $i \neq \ell$, we have
	$|\tilde{D}_e \cap V(\mathcal{S}_{i})| \ge 1$, since $\tilde{D}_e$ dominates all the vertices $v_i$.
	
	It is clear that the  vertices $u_\ell$ and $\bar{u}_\ell$ do not simultaneously appear in the same clause in 
	$\mathcal{C}$, so, there is no $j$ such that the vertex $c_j$ is adjacent to both of them. 
	Since $u_\ell$ and $\bar{u}_\ell$ should be dominated by $\tilde{D}_e$, then there exists two distinct vertices 
	$c_j,c_\ell \in \tilde{D}_e$ such that $c_j$  and $c_\ell$ dominate   $u_\ell$ and $\bar{u}_\ell$, respectively.
	
	Hence, $|\tilde{D}_e| \ge n+1$, which is a contradiction. Therefore, we have 	$|\tilde{D}_e \cap V(\mathcal{S}_{i})| = 1$ 
	for all $i=1,2,\dots,n$ and 	$c_j \notin \tilde{D}_e$ for every $j$, since $|\tilde{D}_e|=n$.
	
	\end {proof}

\begin{claim}
	\label{clim33}
	The 3-SAT instance $\mathcal{(U,C)}$ is satisfiable if and only if $r_{OCD}(G) = 1$.
\end{claim}	
\begin{proof}
	Suppose that $r_{OCD} = 1$, which  means that there exists an edge $e$ in $\bar{G}$ such that
		$\tilde{\gamma}_c(G + e) = n$.
		 Let $ \tilde{D_e}$ be a  $\tilde{\gamma}$-set   of  $G+e$.		 
		 Then, by Claim \ref{claim32}, for all $i=1,2,\dots,n$, we have $|\tilde{D}_e \cap V(\mathcal{S}_i)| =1$. To be precise, we have either  $\tilde{D}_e \cap V(\mathcal{S}_i) = \{v_i\} $, $\tilde{D}_e \cap V(\mathcal{S}_i) =\{u_i\}$ or $\tilde{D}_e \cap V(\mathcal{S}_i) =\{\bar{u_i}\}$ for all $i=1,2,\dots,n$.
		
		Assume that the  mapping $f : \mathcal{U} \rightarrow \{T,F\} $ is defined as


		\begin{equation}
	  \label{e3}
		 f(u_i)=  \begin{cases} T, & \text{if} ~~ u_i \in \tilde{D}_e~~  \text{or} ~~ v \in \tilde{D}_e,\\  F, & \text{if} ~~ \bar{u_i} \in \tilde{D}_e. 
		   \end{cases}
		\end{equation}
	
	We want to show that the mapping $f$ is a satisfying truth assignment for $I = (\mathcal{U,C})$. So, it is sufficient to show that $f$ satisfies every clause in $\mathcal{C}$. We choose an arbitrary  clause  $C_j \in \mathcal{C}$. Since the corresponding vertex $c_j$ to clause $C_j$ is not adjacent to any vertices in correspondence with the set $\{v_i : 1 \le i \le n\}$, there exists an index $i$
	such that $c_j$ is dominated by $u_i \in \tilde{D}_e$ or $\bar{u_i} \in \tilde{D}_e$. Assume that $c_j$ is dominated by $u_i \in \tilde{D}_e$, then $u_i$ is adjacent to vertex $c_j$ in $G$, namely, the literal $u_i$ is in the clause $C_j$. Since $u_i \in \tilde{D}_e$, we have $f(u_i) = T $ by Equation \ref{e3}. So, $f$ satisfies the clause $C_j$.\\
	Now, suppose that the vertex $c_j$ is dominated by vertex $\bar{u_i} \in \tilde{D}_e$. So, $\bar{u_i}$ is adjacent to $c_j$ in $G$, namely, the literal $\bar{u_i}$ is in the clause $C_j$. Since $\bar{u_i} \in \tilde{D}_e$, we have $f(u_i) = F $ by Equation \ref{e3}, which implies that $\bar{u_i}$ is assigned the truth value $T$ by $f$. So, the clause $C_j$ is satisfied by $f$. Since clause $C_j$ is chosen arbitrarily, all the clauses in $\mathcal{C}$ are satisfied by $f$, which implies that $I = (\mathcal{U,C})$ is satisfiable.

	Conversly, suppose that $f : \mathcal{U} \rightarrow \{T,F\} $ is a satisfying truth assignment for $\mathcal{C}$
	and $\tilde{D}^\prime$  is a subset of $V(G)$ that is  constructed as follows.
	
	If $f(u_i)  = T$,  then we put the vertex $u_i$ in $\tilde{D}^\prime$ and if 
	$f(u_i) = F $, we put the vertex $\bar{u_i}$ in $\tilde{D}^\prime$.
	 Therefore, we have $|\tilde{D}^\prime| = n$. For  $ j=1,2,\dots,m$, at least one of the literals in clause $C_j$ 
	 is true under the assignment of $f$, given that $f$ is a  satisfying truth assignment for $I = (\mathcal{U,C})$. 
	 So, by the construction of $G$, the corresponding vertex $c_j$ in $G$ is adjacent to at least one vertex in 
	 $\tilde{D}^\prime$.
 Without loss of generality, 
 let $f(u_1)= T $. Then, $\tilde{D}^\prime$ is a dominating set for $G + \{x,u_1\}$. On 
 the other hand, the induced graph $G[V\setminus \tilde{D}^\prime]$ is connected.
 Hence, $\tilde{D}^\prime$ is an outer-connected dominating set for $G+\{x,u_1\}$ and $\tilde{\gamma}_c(G+\{x,u_1\}) \le |\tilde{D}^\prime| = n$.
	
By Claim \ref{claim31},  we have		$\tilde{\gamma}_c(G) = n+1$. Therefore, we obtain $\tilde{\gamma}_c(G+\{x,u_1\}) \le n < n+1 = \tilde{\gamma}_c(G) $, which means that $r_{OCD} = 1$.

\end{proof}
\end{quote}
 Claims \ref{claim31}, \ref{claim32} and \ref{clim33} conclude the proof.

\end {proof}

\section{The $\mathbf{NP}$-hardness of the outer-connected bondage}
In this section, we  show that the outer-connected bondage problem for general graphs is an $\mathbf{NP}$-hard problem. Consider the following decision problem.

Outer-connected bondage problem:

Instance: A graph $G$ with no isolated vertices and a positive integer k.

Question: Is $ b_{OCD}(G) \le k $?

\begin{theorem}
	The outer-connected bondage problem is $\mathbf{NP}$-hard.
	\end{theorem}
	\begin{proof}
		
		Let $I = (\mathcal{U} = \set{u_1,u_2,\dots,u_n}$ , $\mathcal{C} = \set{C_1,C_2,\dots,C_m})$ be an arbitrary instance of
		the 3-SAT problem. For an arbitrary positive integer k,  we will construct a graph $G$  such that this instance of 3-SAT
		will be satisfiable if and only if $G$ has an outer-connected bondage of cardinality of at most k, i.e. $ b_{OCD}(G) \le k $ . 
		The graph $G$ is constructed as follows.
		
		To each $u_i \in \mathcal{U}$, we associate a  vertex set $\mathcal{H}_i= \set{u_i,v_i,\bar{u_i},x_i,y_i}$ and add edges $\{x_i,u_i\}, \{y_i,\bar{u_i}\}, \{u_i,v_i\}$, and $\{\bar{u_i},v_i\} $ for $i = 1,2,\dots,n$.  For each  clause $C_j \in \mathcal{C}$, we associate a single vertex $c_j$ and add edge $\{c_j,u_i\}(\{c_j,\bar{u_i}\})$   if the literal $u_i(\bar{u_i} )$ is present in the clause $C_j$, where  $j = 1,2,\dots,m$. Then, we add a set of vertices $S = \{s_1,s_2,s_3,s_4\} $ and join $s_1,s_3$ and $s_4$ to  vertices $c_j$  and $s_2$. Finally, we add a vertex $t$ to the graph $G$, edges $\{t,s_1\}, \{t,s_3\}, \{t,s_4\}$, $\{t,u_i\}$ and $ \{t,\bar{u_i}\}$  for  $i = 1,2,\dots,n$ and set $k=1$.
		
		\begin{figure}
\centering
\includegraphics[width=0.7\linewidth]{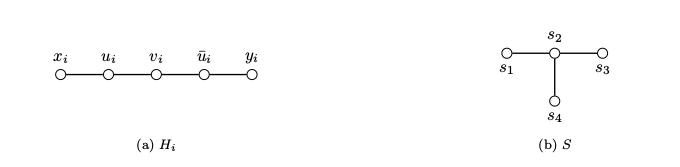}
\caption{{\footnotesize The graphs $H_i$ and $S$}}
\label{fig2}
\end{figure}

%
%

		It can be  easily seen that the construction can be accomplished in polynomial time, since the graph $G$ contains $5n+m+5$ vertices and $6m+6n+6$ edges.  All that remains to be shown is 
		that $I = \mathcal{(U,C)}$ is satisfiable if and only if $b_{OCD}(G) \le k$. Without loss of generality, let $k=1$. To this end,
		we will first prove the following five claims.
		
		\begin{quote}	
			\begin{claim} 
				\label{claim51}
				
				For any graph $G$ constructed as above, we have $\tilde{\gamma}_c(G) \ge 3n+1$.
				\end{claim}
				
				\begin{proof}
					Let $\tilde{D}$ be a 	$\tilde{\gamma}$-set of $G$. Then, 	$\tilde{\gamma}_c(G) = |\tilde{D}| \ge 3n+1$, 
					since $|\tilde{D} \cap V(\mathcal{H}_i)| \ge 3$ for  $i=1,2,\dots,n$. Note that to dominate a vertex $v_i $,  we need  at least one vertex and the leaf vertices $x_i$ and $y_i$ to be in $\tilde{D}$. Moreover, $|\tilde{D} \cap N[s_2]| \ge 1$. 
					\end{proof}	
					
					\begin{claim} 
						\label{claim52}
						If  $\tilde{\gamma}_c(G) = 3n+1$, then $c_j , t \notin \tilde{D}$ for  $j=1,2,\dots,m$, $\tilde{D} \cap V(S) = \{s_2\}$ and 	$|\tilde{D} \cap V(\mathcal{H}_i)| = 3$ for  $i=1,2,\dots,n$.
						\end{claim}		 	
						
						\begin{proof}
							Since the connection between  $\mathcal{H}_i$  and $S$ is due to the vertex $t,$ then 	$t \notin \tilde{D}$. Suppose that  $\tilde{\gamma}_c(G) = 3n+1$. Then,  
							$|\tilde{D} \cap V(\mathcal{H}_i)| = 3$ for  $i=1,2,\dots,n$, while $|\tilde{D} \cap V(S)| = 1$. Consequently, $c_j  \notin \tilde{D}$ for  $j=1,2,\dots,m$. Simultaneously, if $\tilde{D} \cap V(S) = \{s_1\}$, then, $s_3$ and $s_4$ are not dominated. Hence, $s_1 \notin \tilde{D}$ and similary, $s_3, s_4 \notin \tilde{D}$. So, $\tilde{D} \cap V(S) = \{s_2\}$.
							\end{proof}

							\begin{claim}
								\label{claim53}
								The 3-SAT instance $I = \mathcal{(U,C)}$ is satisfiable if and only if $\tilde{\gamma}_c(G) = 3n+1$.
								\end{claim}	
								\begin{proof}
									Suppose that $\tilde{\gamma}_c(G) = 3n+1$ and $c_j$ is   an arbitrary vertex. By Claim \ref{claim52}, this vertex is  adjacent to  either $u_i \in \tilde{D}$ or $\bar{u_i} \in \tilde{D}$,	since $ s_1 , s_3, s_4 \notin \tilde{D}$. As 	$|\tilde{D} \cap V(\mathcal{H}_i)| = 3$ for  $i=1,2,\dots,n$, it follows that either	$\tilde{D} \cap V(\mathcal{H}_i) = \{x_i,y_i,u_i\}$, $\tilde{D} \cap V(\mathcal{H}_i) = \{x_i,y_i,\bar{u_i}\}$ or $\tilde{D} \cap V(\mathcal{H}_i) = \{x_i,y_i,v_i\}$.

									Let the mapping $f : \mathcal{U} \rightarrow \{T,F\} $ be defined as
									\begin{equation}
									\label{e5}
									f(u_i)=  \begin{cases} T, & \text{if} ~~ u_i \in \tilde{D} ~~ \text{or} ~~ v_i \in \tilde{D}, \\  F, & \text{if} ~~ \bar{u_i} \in \tilde{D}. 
									\end{cases}
									\end{equation}

										To prove that the values assigned by the  mapping $f$ is a satisfying truth assignment for $I = (\mathcal{U,C})$, it is sufficient to show that $f$ satisfies every clause in $\mathcal{C}$. Let   $C_j \in \mathcal{C}$ be an arbitrarily  clause. Since the corresponding vertex to the clause $C_j$ is not adjacent to any vertex in correspondence with the set $\{v_i,x_i,y_i : 1 \le i \le n\}$, there exists an $i$
									such that $c_j$ is dominated by either $u_i \in \tilde{D}$ or $\bar{u_i} \in \tilde{D}$. Without loss of generality, assume that $c_j$ is dominated by $u_i \in \tilde{D}$. So, $u_i$ is adjacent to $c_j$ in $G$, namely the literal $u_i$ is in the clause $C_j$. Since $u_i \in \tilde{D}$, we have $f(u_i) = T $ by Equation \ref{e5}. So, the values assigned by the mapping $f$ satisfies the clause $C_j$.\\
									Now, suppose that the vertex $c_j$ is dominated by vertex $\bar{u_i} \in \tilde{D}$. So, the vertex $\bar{u_i}$ is adjacent to the vertex $c_j$ in $G$, namely the literal $\bar{u_i}$ is in the clause $C_j$. Since $\bar{u_i} \in \tilde{D}$, we have $f(u_i) = F $ by Equation \ref{e5}, which implies that $\bar{u_i}$ is assigned the truth value $T$ by $f$ and the clause $C_j$ is satisfied by $f$. Since $C_j$ was chosen arbitrarily, all the clauses in $\mathcal{C}$ are satisfied by $f$, which implies that $I = (\mathcal{U,C})$ is satisfiable.	
									
									Conversly, suppose that $f : \mathcal{U} \rightarrow \{T,F\} $ is a satisfying truth assignment for $\mathcal{C}$
									and $\tilde{D}^\prime$  is a subset of $V(G)$ which is  constructed as follows.	
									If $f(u_i)  = T$, then we put the vertex $u_i$ in $\tilde{D}^\prime$ and if 
									$f(u_i) = F $, then we put the vertex $\bar{u_i}$ in $\tilde{D}^\prime$.
									Therefore, we have $|\tilde{D}^\prime| = n$. For  $ j=1,2,\dots,m$, at least one of the literals in $C_j$ 
									is true under the assignment $f$, because the mapping $f$ is a  satisfying truth assignment for $I = (\mathcal{U,C})$. 
									So, by the construction of $G$, the vertex is in  correspondence to $C_j$ in $G$ is adjacent to at least one vertex in 
									$\tilde{D}^\prime$.
									Then, $D = \tilde{D}^\prime \cup (\bigcup_{i=1}^{n}{\{x_i,y_i\}} ) \cup \{s_2\}$ is a dominating set for $G$. On 
									the other hand, the induced graph $G[V\setminus \ D]$ is connected.
									Hence, $D$ is an outer-connected dominating set for $G$ and $\tilde{\gamma}_c(G) \le |D| = 3n + 1$. 	
									By Claim \ref{claim51},  we have		$\tilde{\gamma}_c(G) \ge 3n+1$. Therefore, we obtain $\tilde{\gamma}_c(G) = 3n+1$. 
									
									\end{proof}
									
									\begin{claim}
										\label{claim54}
										For every $e \in E(G)$, we have  $\tilde{\gamma}_c(G-e) \le 3n+2$.
										\end{claim}
										\begin{proof}
											Suppose that $E^\prime = \{\{s_2,s_3\},\{s_2,s_4\},\{s_1,c_j\},\{u_i,v_i\},\{y_i,\bar{u_i}\}, \{t,s_1\},\{v_i,\bar{u_i}\},\{t,\bar{u_i}\}\}$ and $E^{\prime\prime} =E \setminus E^\prime$. Let  $e \in E^{\prime\prime}$ be  an edge.  It is clear that the set $ D^\prime = (\bigcup_{i=1}^{n}{\{x_i,y_i,u_i\}}) \cup \{s_1,s_2\}$ is an outer-connected dominating set for $G-e$,
											since every vertex in $V \setminus D^\prime$ is adjacent to a vertex in $D^\prime$ due to an edge in $E^\prime$, and the induced graph $(G-e)[V\setminus D^\prime]$ is connected. This connection is estsblished by vertices $t$ and $s_i$ for $i \ne 1,2$. Given that $|D ^\prime|= 3n+2$,  then $\tilde{\gamma}_c(G-e) \le 3n+2$. We have four cases to consider:
											
									 \item[Case 1:] 	If either $e = \{s_2,s_3\} $, $e= \{s_1,c_j\}$ or $e = \{t,s_1\}$, then						 
									   $ D^\prime = (\bigcup_{i=1}^{n}{\{x_i,y_i,u_i\}}) \cup \{s_3,s_2\}$ is an outer-connected dominating set for $G-e$ and 
									 $\tilde{\gamma}_c(G-e) \le |D^\prime|= 3n+2$.

									\item[Case 2:]		If $e = \{s_2,s_4\} $, then $ D^\prime = (\bigcup_{i=1}^{n}{\{x_i,y_i,u_i\}}) \cup \{s_4,s_2\}$ is an outer-connected dominating set for $G-e$ and 
											$\tilde{\gamma}_c(G-e) \le |D^\prime|= 3n+2$. 
											
									\item[Case 3:]		If either $e = \{y_i,\bar{u_i}\} $, $e= \{\bar{u_i},v_i\}$ or $e= \{u_i,v_i\}$, then $ D^\prime = (\bigcup_{i=1}^{n}{\{x_i,y_i,v_i\}}) \cup \{s_1,s_2\}$ is an outer-connected dominating set for $G-e$ and 
											$\tilde{\gamma}_c(G-e) \le |D^\prime|= 3n+2$.
											
									\item[Case 4:]		If $e = \{t,\bar{u_i}\} $, then $ D^\prime = (\bigcup_{i=1}^{n}{\{x_i,y_i,\bar{u_i}\}}) \cup \{s_1,s_2\}$ is an outer-connected dominating set for $G-e$ and 
									$\tilde{\gamma}_c(G-e) \le |D^\prime|= 3n+2$.

											\end{proof}		
											\begin{claim}
												\label{claim55}
												$\tilde{\gamma}_c(G) = 3n+1 $ if and only if $ b_{OCD}(G) =1 $.
												\end{claim}	
												
												\begin{proof}
													First, suppose that $\tilde{\gamma}_c(G) = 3n+1 $. Let $e=\{s_1,s_2\}$ and  $\tilde{\gamma}_c(G) =\tilde{\gamma}_c(G-e)$. If $\tilde{D}$ is a $\tilde{\gamma}$-set of $G-e$, then  $\tilde{D}$ is a $\tilde{\gamma}$-set for $G$ of cardinality $3n+1$. By Claim \ref{claim52}, we have $c_j,  t \notin \tilde{D}$ for  $j=1,2,\dots,m$ and $\tilde{D} \cap V(S) = \{s_2\}$. So, the vertex $s_1$ is not dominated by $\tilde{D}$, which is a contradiction. Hence, $\tilde{\gamma}_c(G) < \tilde{\gamma}_c(G-e)$. So, $ b_{OCD}(G) =1 $.
													
													Next, assume that $ b_{OCD}(G) =1 $. By Claim \ref{claim51}, it follows that $\tilde{\gamma}_c(G) \ge 3n+1 $. Suppose that $e$ is an edge such that $\tilde{\gamma}_c(G) < \tilde{\gamma}_c(G-e)$. By Claim \ref{claim54}, we have $3n+1 \le \tilde{\gamma}_c(G) < \tilde{\gamma}_c(G-e) \le 3n+2$ which implies that $\tilde{\gamma}_c(G)= 3n+1$.

													\end{proof}
													
													Therefore, by Claims \ref{claim53} and  \ref{claim55}, we have $ b_{OCD}(G) =1 $ if and only if $I = \mathcal{(U,C)}$ is satisfiable. 
													\end{quote}
													\end{proof}

\section{Exact values for $b_{OCD}(G)$}
In this section, we establish several theorems on the exact values of $ b_{OCD}(G)$.
\begin{lemma}
	\label{lem1}
	Let $G=K_n$ be  a complete graph with $n \ge 3$ vertices. If  $\lceil \frac{n+1}{2} \rceil -1$ edges are removed from $G$, then $G$ is still
	 connected. 
\end{lemma}
\begin{proof}
	Two cases are considered:
	\begin{enumerate}
		\item If  $n$ is odd, then $\lceil \frac{n+1}{2} \rceil = \frac{n+1}{2}$ and the number of remaining edges by removing $\lceil \frac{n+1}{2} \rceil -1$ edges  from $G$ is equal to 
		\begin{equation}
		   \frac{n(n-1)}{2} - (\frac{n+1}{2} -1) = \frac {(n-1)^2}{2}. 
		\end{equation}
	On the other hand, a graph of order $n$ cannot be disconnected if it is simple and the number of its edges is greater than $\frac{(n-1)(n-2)}{2}$. This is because the complete graph on $n-1$ vertices has $\frac{(n-1)(n-2)}{2}$ edges. Consequently, since $\frac {(n-1)^2}{2}>\frac{(n-1)(n-2)}{2}$ and $n \ge 3$, the graph $G$ remains connected even after the removal of  $\lceil \frac{n+1}{2} \rceil -1$ edges.
	
		\item If  $n$ is even, then $\lceil \frac{n+1}{2} \rceil = \frac{n+2}{2}$ and the number of remaining edges by removing $\lceil \frac{n+1}{2} \rceil -1$ edges  from $G$ is equal 
		\begin{equation}
		\frac{n(n-1)}{2} - (\frac{n+2}{2} -1) = \frac {n(n-2)}{2}. 
		\end{equation}
		On the other hand, a graph of order $n$ cannot be disconnected if it is simple and the number of its 
		edges is greater than $\frac{(n-1)(n-2)}{2}$. This is because the complete graph with $n-1$ vertices has $\frac{(n-1)(n-2)}{2}$ edges. Consequently, since $\frac {n(n-2)}{2}>\frac{(n-1)(n-2)}{2}$ and $n \ge 3$, the graph $G$ remains connected even after the removal of  $\lceil \frac{n+1}{2} \rceil -1$ edges.	
			
	\end{enumerate}
\end{proof}
\begin{theorem}
	Let $K_n$ be  a complete graph with $n \ge 3$ vertices. Then, we have
	  	\begin{equation}
	  	b_{OCD}(K_n)= \begin{cases}	1, & \text{if} ~~ n=3, \\ \lceil \frac{n}{2} \rceil, & \text{otherwise}. 
	  	\end{cases}
	  	\end{equation}
	  	
\end{theorem}

\begin{proof}
	If $n=3$,  then we have  $\tilde{\gamma}_c(K_3) = 1$. By removing any edges from $G$, it turns into a $P_3$. So, we have $ b_{OCD}(K_3)=1 $ since $\tilde{\gamma}_c(P_3) = 2$. 
	
	Now, suppose that $n > 3$. Let the  graph $G^ \prime$ be obtained by removing fewer than $\lceil \frac{n}{2} \rceil$ edges from $G$. Then, $G^ \prime$ contains at least a vertex of degree $n-1$ . Let the vertex $v$ be of degree $n-1$. In the other hand,  according to the Lemma \ref{lem1} we have the induced graph $G^ \prime[V \setminus {v}]$ is connected. Then, $\tilde{\gamma}_c(G^ \prime) =1$.  So, we have 
	\begin{equation}
	\label{e11}
	b_{OCD}(G) \ge \lceil \frac{n}{2} \rceil.
	\end{equation} 
	Now, we need to consider the following two cases:
	\begin{enumerate}
		\item	 If  $n$ is even. Let $H$  be the graph obtained by removing $\lceil \frac{n}{2} \rceil$ independent edges from $G$. Then the degree of every  vertex $v \in V(H)$ is $n-2$. So, we have $\tilde{\gamma}_c(H) \ge 2$.
		\item    If $n$ is odd. Let $H$  be the graph obtained by removing $ \frac{n-1}{2}$ independent edges from $G$. Then, there is exactly one  vertex $v \in V(H)$ such that the degree of $v$ is  $n-1$. If we remove one edge incident with $v$, then we have  $\tilde{\gamma}_c(H) \ge 2$.
		  \end{enumerate}
	 
	  In either cases, by removing $\lceil \frac{n}{2} \rceil$ we have $\tilde{\gamma}_c(G) < \tilde{\gamma}_c(H)$. So, 
	  \begin{equation}
	  \label{e12}
	  b_{OCD}(G) \le \lceil \frac{n}{2} \rceil.
	  \end{equation}
	 
	 Therefore, by Equations \ref{e11} and \ref{e12}, we have $b_{OCD}(G) = \lceil \frac{n}{2} \rceil$.
	   
\end{proof}


\begin{theorem}
	\label{t71}
	Let  $C_n$ be  a cycle graph with $n \ge 3$ vertices. Then, we have
	\begin{equation}
	b_{OCD}(C_n)= 		  
	\begin{cases}
	1, &  \text{if} ~~~ n=3, \\ 
	\lceil \frac{n}{3} \rceil, &  \text{otherwise}. 
	\end{cases}
	\end{equation}
	
\end{theorem}

\begin{proof}
	If $n=3$, then we have $\tilde{\gamma}_c(C_3) = 1$. By removing any edges from $G$, it turns into a $P_3$. So, $ b_{OCD}(C_3)=1 $ since $\tilde{\gamma}_c(P_3) = 2$.
	
	Cyman in \cite{cyman2007outer1connected} has shown that if $n \ge 4$, then $\tilde{\gamma}_c(C_n) = \tilde{\gamma}_c(P_n) = n-2$. So, for $n \ge 4$, we have $ b_{OCD}(C_3) \ge 1 $. By removing some of the edges from $P_n$, a set of components $\{H_1,H_2,\cdots,H_m\}$ is obtained such that every component is a path. If
	$H = \bigcup_{i=1}^{m}{H_i} $ and $\tilde{D}(H)$ is an outer-connected dominating set for $H$, then $\tilde{D}(H) =  \bigcup_{i=1}^{m-1}{H_i} \cup \tilde{D}(H_m)$ (See Lemma 3.1 in \cite{hashemipour2017outer}). 
	Therefore, if there exists  at least one component such as $H_i$ with four vertices, then $\tilde{\gamma}_c(C_n) = \tilde{\gamma}_c(P_n) = \tilde{\gamma}_c(H) =n-2$. Otherwise, we have $\tilde{\gamma}_c(C_n) = \tilde{\gamma}_c(P_n) < \tilde{\gamma}_c(H) =n-1$. So, we need to break the path $P_n$ in such a way that there exist no more than three vertices in any component $H_i$. To this end, we have to remove $\lceil \frac{n}{3} \rceil - 1$ edges from $P_n$, which with the deleted edge from $C_n$ are counted to $\lceil \frac{n}{3} \rceil $. 
\end{proof}

\begin{theorem}
	Let  $P_n$ be a path with $n \ge 3$ vertices. Then, we have
	\begin{equation}
	b_{OCD}(P_n)= 		  
	\begin{cases}
	1, & \text{if} ~~ n=2, \\
	2, & \text{if} ~~ n=3, \\ 
	\lceil \frac{n}{3} \rceil - 1{\tiny }, & \text{otherwise}. 
	\end{cases}
	\end{equation}
	
\end{theorem}

\begin{proof}
The cases for $n=2$ and $n=3$  are quite clear. The proof of the case $n \ge 4 $ is the same as in Theorem \ref{t71}.	
\end{proof}

\begin{theorem}
	The graph $G$ is a galexy of order $n \ge 4$ if and only if $b_{OCD}(G)=|E(G)|$.
\end{theorem}
\begin{proof}
	According to  the Observation 2 in \cite{cyman2007outer1connected} and Lemma 3.1 in \cite{hashemipour2017outer}, it is clear that if $G$ is a galexy, then  $\tilde{\gamma}_c(G) = n-1$. Moreover, the only  graph with   outer-connected domination number equal to $n$ is $\bar{K_n}$. So, $b_{OCD}(G)=|E(G)|$.
	Converesly, suppose that $b_{OCD}(G)=|E(G)|$. If a component $H$ of $G$ is not a star, then $H$ either contains a cycle or a $P_4$, which means that  $\tilde{\gamma}_c(G) \le n-2$. Let $e=\{v_1,v_2\}$ be an edge in the cycle or the $P_4$. If $H$ is a graph obtained  by removing all  edges from $G$ expect $e$, then we have $\tilde{\gamma}_c(H) = n-1 > \tilde{\gamma}_c(G) $. This implies $b_{OCD}(G) \le |E(G)|-1$, which is a contradiction. 
\end{proof}

%

\end{document}